\documentclass[a4paper,10pt]{article}
\usepackage{charter,eulervm}

\usepackage{marvosym}
\usepackage{amssymb,amsmath,amsthm}
\usepackage{graphicx}
\usepackage[usenames,dvipsnames]{color}
\usepackage[pdftex,breaklinks,colorlinks,
    citecolor={blue}, linkcolor={blue},urlcolor=Maroon]{hyperref}
\usepackage{multirow,booktabs,array}
\usepackage{paralist}

\begin{document}

\newtheorem{theorem}{Theorem}[section]
\newtheorem{lemma}[theorem]{Lemma}
\newtheorem{corollary}[theorem]{Corollary}
\newtheorem{proposition}[theorem]{Proposition}
\newtheorem{exrule}{Exchange Rule}
\newcommand{\blackslug}{\penalty 1000\hbox{
    \vrule height 8pt width .4pt\hskip -.4pt
    \vbox{\hrule width 8pt height .4pt\vskip -.4pt
          \vskip 8pt
      \vskip -.4pt\hrule width 8pt height .4pt}
    \hskip -3.9pt
    \vrule height 8pt width .4pt}}
\newtheorem{definition}{Definition}
\def\boxit#1{\vbox{\hrule\hbox{\vrule\kern4pt
  \vbox{\kern1pt#1\kern1pt}
\kern2pt\vrule}\hrule}}

\newcommand{\ist}{internal spanning tree}
\newcommand{\iob}{internal out-branching}
\newcommand{\opt}[1]{\ensuremath{{\mathtt{opt}(#1)}}}
\addtolength{\baselineskip}{+0.4mm}

\title{\vspace{-.4in} \bf A $2k$-Vertex Kernel for Maximum Internal
  Spanning Tree\thanks{Supported by the National Natural
    Science Foundation of China under grants 61232001, 61472449, and
    61420106009.}}

\author{Wenjun Li\thanks{School of Information Science and
    Engineering, Central South University, Changsha, China.}
  \and
  \addtocounter{footnote}{-1} Jianxin Wang\footnotemark 
  \and
  \addtocounter{footnote}{-1} Jianer Chen\footnotemark\
  \thanks{Department of Computer Science and Engineering, Texas A\&M
    University, College Station, Texas.}
  \and
  Yixin Cao$^{\text{\Letter}}$\thanks{Department of Computing, Hong
    Kong Polytechnic University, Hong Kong,
    China. \href{mailto:yixin.cao@polyu.edu.hk} {\tt
      yixin.cao@polyu.edu.hk}.}  }
\date{}
\maketitle

 \begin{abstract}
   We consider the parameterized version of the maximum internal
   spanning tree problem, which, given an $n$-vertex graph and a
   parameter $k$, asks for a spanning tree with at least $k$ internal
   vertices.  Fomin et al.~[J. Comput. System Sci., 79:1--6] crafted a
   very ingenious reduction rule, and showed that a simple application
   of this rule is sufficient to yield a $3k$-vertex kernel.  Here we
   propose a novel way to use the same reduction rule, resulting in an
   improved $2k$-vertex kernel.  Our algorithm applies first a greedy
   procedure consisting of a sequence of local exchange operations,
   which ends with a local-optimal spanning tree, and then uses this
   special tree to find a reducible structure.  As a corollary of our
   kernel, we obtain a deterministic algorithm for the problem running
   in time $4^k \cdot n^{O(1)}$.
\end{abstract}

\section{Introduction}
A {\em spanning tree} of a connected graph $G$ is a subgraph that
includes all the vertices of $G$ and is a tree.  Spanning tree is a
fundamental concept in graph theory, and finding a spanning tree of
the input graph is a routine step of graph algorithms, though it
usually induces no extra cost: most algorithms will start from
exploring the input graph anyway, and both breadth- and
depth-first-search procedures produce a spanning tree as a byproduct.
However, a graph can have an exponential number of spanning trees, of
which some might suit a specific application better than others.  We
are hence asked to find constrained spanning trees, i.e., spanning
trees minimizing or maximizing certain objective functions.  The most
classic example is the minimum-weight spanning tree problem (in
weighted graphs), which has an equivalent but less known formulation,
i.e., maximum-weight spanning tree.  Other constraints that have
received wide attention include minimum diameter spanning tree
\cite{hassin-95-minimum-diameter-spanning-tree}, degree constrained
spanning tree~\cite{konemann-02-degree-bounded-mst,
  konemann-05-degree-bounded-mst}, maximum leaf spanning tree
\cite{lu-98-approximate-maximum-leaf-spanning-tree}, and maximum
internal spanning tree \cite{salamon-10-thesis-mist}.  Unlike the
minimum-weight spanning tree problem
\cite{karger-95-randomized-linear-mst}, most of these constrained
versions are NP-hard \cite{ozeki-11-survey-spanning-trees}.

The optimization objective we consider here is to maximize the number
of internal vertices (i.e., non-leaf vertices) of the spanning tree,
or equivalently, to minimize the number of its leaves.  More formally,
the {\em maximum internal spanning tree} problem asks whether a given
graph $G$ has a spanning tree with at least $k$ internal vertices.
Containing the Hamiltonian path problem as a special case, it is
clearly NP-hard.  This paper approaches it by studying {kernelization
  algorithms} for its parameterized version; here the parameter is
$k$, and hence we use the name {\em $k$-\ist}.  Given an instance $(G,
k)$ of $k$-\ist, a {\em kernelization algorithm} produces in
polynomial time an ``equivalent'' instance $(G', k')$ such that $k'
\leq k$ and that the {\em kernel size} (i.e., the number of vertices
in $G'$) is upper bounded by some function of $k'$.  Prieto and Sloper
\cite {prieto-03-mist} presented an $O(k^3)$-vertex kernel for the
problem, and improved it to $O(k^2)$ in the journal version \cite
{prieto-05-mist}.  Fomin et al.~\cite{fomin-13-kernel-mist} crafted a
very ingenious reduction rule, and showed that a simple application of
this rule is sufficient to yield a $3k$-vertex kernel.  Answering a
question asked by Fomin et al.~\cite{fomin-13-kernel-mist}, we further
improve the kernel size to $2k$.
\begin{theorem}\label{thm:main-2}
  The $k$-\ist\ problem has a $2k$-vertex kernel.
\end{theorem}

We obtain this improved result by revisiting the reduction rule
proposed by Fomin et al.~\cite{fomin-13-kernel-mist}.  A nonempty
independent set $X$ (i.e., a subset of vertices that are pairwise
nonadjacent in $G$) as well as its neighborhood are called a {\em
  reducible structure} if $|X|$ is at least twice as the cardinality
of its neighborhood.  To apply the reduction rule one needs a
reducible structure.  Indeed, we are proving a stronger statement that
implies Theorem~\ref{thm:main-2} as a corollary.
\begin{theorem}\label{thm:main}
  Given an $n$-vertex graph $G$, we can find in polynomial time either
  a spanning tree of $G$ with at least $n /2$ internal vertices, or a
  reducible structure.
\end{theorem}

The observation in \cite{fomin-13-kernel-mist} is that the leaves of a
depth-first-search tree $T$ are independent.  Therefore, if the graph
has more than $3 k - 3$ vertices, then either the problem has been
solved (when $T$ has $k$ or more internal vertices), or the set of (at
least $2k - 2$) leaves of $T$ will be the required independent set.
It is, however, very nontrivial to find a reducible structure when $2k
< n < 3 k - 3$, and this will be the focus of this paper.  We first
preprocess the tree $T$ using a greedy procedure that applies a
sequence of local-exchange operations to increase the number of its
internal vertices.  After a local optimal spanning tree is obtained,
we show that if it has more leaves than internal vertices, then a
subset of its leaves and its neighborhood make the reducible
structure.  We apply the reduction rule of \cite{fomin-13-kernel-mist}
to reduce it and then repeat the process, which terminates on either a
$2k$-vertex kernel or a solution.

It is interesting to point out that our kernelization algorithm will
never ends with a NO situation, which is common in kernelization
algorithms in literature.  It either returns a trivial YES instance,
or continuously reduces the graph until it has a spanning tree with at
least half internal vertices.  This also means that our kernelization
algorithm does not rely on the parameter $k$.  One should be noted
that to further improve a $2k$-vertex kernel for a graph problem seems
to be a very challenging, if possible, undertaking: more and more such
kernels have appeared in literature, which have stubbornly withstood
all subsequent attacks, however hard they were.

Priesto and Sloper~\cite{prieto-05-mist} also initiated the study of
parameterized algorithms (i.e., algorithms running in time
$O(f(k)\cdot n^{O(1)})$ for some function $f$ independent of
$n$)\footnote{Following convention, we use the $O^*(f(k))$ notation to
  suppress the polynomial factor $n^{O(1)}$ in the running time.}  for
$k$-\ist, which have undergone a sequence of improvement.  Closely
related here is to the {\em $k$-\iob} problem, which, given a {\em
  directed} graph $G$ and a parameter $k$, asks if $G$ has an
out-branching (i.e., a spanning tree having exactly one vertex of
in-degree $0$) with at least $k$ vertices of positive out-degrees.  It
is known that any $O^*(f(k))$-time algorithm for $k$-\iob\ can solve
$k$-\ist\ in the same time, but not necessarily the other way.  After
a successive sequence of studies \cite{gutin-09-iob, cohen-10-iob,
  fomin-12-sharp-separation, shachnai-14-representative-families,
  daligault-11-dissertation}, the current best deterministic and
randomized parameterized algorithms for $k$-\iob\ run in time
$O^*(6.86^k)$ and $O^*(4^k)$ respectively, which are also the best
known for $k$-\ist.  Table 1 summarizes the history of this line of
research.

\begin{table}[h]
  \caption{Known parameterized algorithms for problems $k$-\iob\ and  $k$-\ist\
    (note that an algorithm for the former  applies to the later as well).}
  \medskip
  \centering
\begin{tabular}{l |l r r}
  \toprule
  Problem & Running time & Reference & Remark
  \\
  \midrule
  & $O^*(k^{O(k)})$ & Gutin et al.~\cite{gutin-09-iob} &
  \\
  $k$-internal  & $O^*(55.8^k)$ & Cohen et al.~\cite{cohen-10-iob} &
  \\
  out-branching & $O^*(16^{k+o(k)})$ & Fomin et 
  al.~\cite{fomin-12-sharp-separation} &
  \\
  & $O^*(4^k)$ & Daligault and Kim \cite{daligault-11-dissertation} & 
  \emph{randomized}
  \\
  & $O^*(6.86^k)$ & Shachnai and  
  Zehavi~\cite{shachnai-14-representative-families} &  
  \\
  \midrule  
  & $O^*(k^{2.5 k})$ & Priesto and 
  Sloper~\cite{prieto-05-mist} &
  \\
  $k$-internal & $O^*(2.14^k)$ & Binkele-Raible et 
  al.~\cite{binkele-raible-13-mist} &   \emph{cubic graphs}
  \\
  spanning tree  & $O^*(8^k)$ & Fomin et al.~\cite{fomin-13-kernel-mist} &
  \\
  & $O^*(4^k)$ & {This paper} &
  \\
  \bottomrule 
\end{tabular}
\end{table}

The $O^*(4^k)$-time {\em randomized} algorithm for $k$-\iob\
\cite[Theorem 180]{daligault-11-dissertation} was obtained using a
famous algebraic technique developed by Koutis and Williams
\cite{koutis-09-group-algebras-for-fpt}, which, however, is very
unlikely to be derandomized.  As a corollary of
Theorem~\ref{thm:main-2}, we obtain an $O^*(4^k)$-time deterministic
algorithm for $k$-\ist,---it suffices to apply the $O^*(2^n)$-time
algorithm of Nederlof \cite{nederlof-13-polynomial-space-by-PIE} to
the $2k$-vertex kernel produced by
Theorem~\ref{thm:main-2},---matching the running time of the best
randomized algorithm

\begin{theorem}\label{thm:alg}
  The $k$-\ist\ problem can be solved in time $O^*(4^k)$.
\end{theorem}
It remains an open problem to develop a deterministic $O^*(4^k)$-time
algorithm for $k$-\iob.  Note that the minimum spanning tree problem
has been long known to be solvable in randomized linear time
\cite{karger-95-randomized-linear-mst}, while a deterministic
linear-time algorithm is still elusive.  As a final remark, there is
also a line of research devoted to developing approximation algorithms
for maximum \ist\ \cite{prieto-03-mist, knauer-09-approximate-mist,
  salamon-09-approximate-mist, salamon-08-mist}.  In a companion paper
\cite{li-14-approximate-mist}, we have used a similar local-exchange
procedure to improve the approximation ratio to $1.5$.

\section{A greedy local search procedure}

All graphs discussed in this paper shall always be undirected and
simple, and the input graph is assumed to be connected.  The vertex
set and edge set of a graph $G$ are denoted by $V(G)$ and $E(G)$
respectively.  For a vertex $v\in V(G)$, let $N_G(v)$ denote the
neighborhood of $v$ in $G$, and let $d_G(v) := |N_G(v)|$ be its degree
in $G$.  The neighborhood of a subset $U\subseteq V(G)$ of vertices is
defined to be $N_G(U) = \bigcup _{v \in U} N_G(v) \backslash U$.  A
tree $T$ is a {\em spanning tree} of a graph $G$ if $V(T) = V(G)$ and
$E(T)\subseteq E(G)$; edges not in $T$, i.e., $E(G)\setminus E(T)$,
are {\em cotree edges} of $T$.  A vertex $u\in V(T)$ is a {\em leaf}
of $T$ if $d_T(u)=1$, and an {\em internal vertex} of $T$ otherwise;
let $L(T)$ and $I(T)$ denoted the set of leaves and the set of
internal vertices of $T$ respectively.  An internal vertex $u$ of $T$
is a {\em branchpoint} if $d_T(u)\geq 3$.  Let $I_3(T)$ denote the set
of branchpoints of $T$, and let $I_2(T)$ denote other internal
vertices (having degree $2$ in $T$); the three vertex sets $L(T),
I_2(T),$ and $I_3(T)$ partition $V(T)$.

Since $|I(T)| = |V(T)| - |L(T)|$, to maximize it is equivalent to
minimizing the number of leaves.  Also connecting leaves and internal
vertices, especially branchpoints, of a tree $T$ is the following
elementary fact:
\[
|L(T)| - 2 = \sum_{v\in I(T)} \left( d_T(v) - 2 \right) = \sum_{v\in
  I_3(T)} \left( d_T(v) - 2 \right).
\]
Therefore, informally speaking, we need to decrease the number and
degrees of branchpoints.  We start from an arbitrary spanning tree $T$
of $G$.  We may assume that $T$ is not a path, (as otherwise the
problem has been solved,) to which we apply some local exchanges to
increase the number of internal vertices of $T$.  By a local exchange
we mean replacing an edge in $E(T)$ by a cotree edge of $T$ .  Recall
that for any pair of vertices $u,v$ in a tree $T$, there is a unique
path from $u$ to $v$, denoted by $P_T(u,v)$; if $u v$ is a cotree edge
of $T$, then the length of $P_T(u,v)$ is at least two.  To maintain a
tree, a cotree edge $u v$ can only replace a tree edge in $P_T(u, v)$.
Since our purpose is to eliminate leaves, we will be only concerned
with cotree edges incident to leaves of $T$.  The first exchange rule
is self-explanatory; here the existence of the branchpoint is ensured
by the assumption that $T$ is not a path.

\begin{exrule}[\cite{prieto-03-mist}]\label{rule:1}
  If there is a cotree edge connecting two leaves $l_1$ and $l_2$ of
  $T$, then find an edge $u v$ from $P_T(l_1,l_2)$ such that $u$ is a
  branchpoint, and substitute $l_1 l_2$ for $u v$ in $T$.
\end{exrule}
After an exhaustive application of Rule~\ref{rule:1}, all leaves in
the resulting tree are pairwise nonadjacent.  Henceforth we may assume
that each cotree edge is incident to at least one internal vertex.  We
remark that Fomin et al.~\cite{fomin-13-kernel-mist} achieved this by
using depth-first-search tree at first place.  We can surely use the
same way to get the initial spanning tree, but we still need
Rule~\ref{rule:1}, as later operations of the other exchange rule to
follow may introduce cotree edges connecting leaves that are
originally not.

\begin{definition}
  A cotree edge of $T$ is {\em good} if it connects a leaf $l$ and an
  internal vertex $w$ of $T$.  We say that $l w$ {\em crosses} every
  edge in the path $P_T(l, w)$.
\end{definition}
For notational convenience, when referring to a good cotree edge $l
w$, we always put the leaf $l$ first, and when referring to an edge $u
v$ crossed by it, we always put the vertex closer to $l$ first; hence,
$P_T(l,w)$ can be written as $l \cdots u v \cdots w$.  We would like
to point out that the same edge $u v$ can be crossed by two different
cotree edges and they may be referred to by different orders.

Let us consider the impact of a substitution on the involved vertices.
We will avoid the tree edge incident to $l$ but we do allow $v = w$,
and hence there are either three of four vertices involved.  The two
vertices of the cotree edge are clear: $l$ always becomes an internal
vertex, and $w$ always remains internal (independent of $w = v$ or
not).  On the other hand, $u$ and $v$ will remain internal if they are
in $I_3(T)$, but one or both of them may become leaves if they are in
$I_2(T)$ (with the only exception $v = w$).  Although some operation
does not increase, or even decreases, the number of internal vertices,
by switching vertices in $L(T)$ and $I(T)$, it may introduce cotree
edge(s) between leaves in the new tree, which enable us to
subsequently apply Rule~\ref{rule:1} and serve our purpose.  This is
never the case for the first edge in $P_T(u, v)$ and hence we avoid
it.  This observation is formalized in the next reduction rule, for
which we need a technical definition that characterize those vertices
in $I_2(T)$ that can participate in the aforementioned successive
exchanges.

\begin{definition} 
  All vertices in $I_3(T)$ are {\em detachable}.  A vertex $w\in
  I_2(T)$ is {\em detachable} if there exists a good cotree edge $l w$
  of $T$ satisfying at least one of the following: 
  \begin{enumerate}[(1)]
  \item $P_T(l,w)$ visits a branchpoint, or
  \item some vertex $v$ in $P_T(l,w)$ is incident a good cotree edge
    $l'v$ of $T$ with  $l' \ne l$.
  \end{enumerate}
\end{definition}
Let $D(T)$ denote the set of detachable vertices of $T$, and let
$D_2(T)$ denote those detachable vertices in $I_2(T)$.  Then $ D(T) =
D_2(T) \cup I_3(T) \subseteq I(T)$.  Note that the vertex $v$ in
item (2) of the definition of $D_2(T)$ necessarily has degree two, and
possibly $v = w$.

\begin{exrule}\label{rule:2}
  Let $u v$ be an edge crossed by a good cotree edge $l w$ of $T$.
  Substitute $l w$ for $u v$ if any of the following is true.
  \begin{itemize}
  \item $u\in I_3(T)$, and
    \begin{enumerate}[(a)]
    \item\label{case2} $v = w$ or $v\in I_3(T)$; or
    \item\label{case3} $v\in D_2(T)$ and there is a good cotree edge
      $l_v v$ with $l_v \ne l$.
    \end{enumerate}
  \item $u\in D_2(T)$, there is a good cotree edge $l_u u$ with
    $l_u \ne l$, and
    \begin{enumerate}[(a)]
    \setcounter{enumi}{2}
  \item\label{case4} $v = w$ or $v\in I_3(T)$; or
    \item\label{case5} $v\in D_2(T)$ and there is a good cotree
      edges $l_v v$ with $l_v\not\in \{l, l_u\}$.
    \end{enumerate}
  \end{itemize}
  Moreover, after \eqref{case3} and \eqref{case4}, apply
  Rule~\ref{rule:1} to $l_v v$ and $l_u u$ respectively; after
  \eqref{case5}, apply Rule~\ref{rule:1} to $l_u u$ and then to $l_v
  v$ (if still applicable).
\end{exrule}

One can check in polynomial time whether Rule~\ref{rule:2} (and which
case of it) is applicable.  To show that the whole procedure can be
finished in polynomial time, we need to argue that each invocation
increases the number of internal vertices by at least one.  Note that
Rule~\ref{rule:2} is only applied after Rule~\ref{rule:1} is no longer
applicable.
\begin{lemma}\label{lem0} 
  Applying Rule~\ref{rule:1} or Rule~\ref{rule:2} to a spanning tree
  $T$ of $G$ results in a new spanning tree $T'$ of $G$ satisfying
  $|I(T')| \ge |I(T)| + 1$.
\end{lemma}
\begin{proof}
  In Rule~\ref{rule:1} and Rule~\ref{rule:2}, the replaced edge $u v$
  is in $P_T(l_1, l_2)$ and $P_T(l, w)$ respectively, so the resulting
  subgraph $T'$ must be a spanning tree of $G$.  To compare the number
  of internal vertices, it suffices to consider these vertices
  incident to the added/deleted edges.

  After the application of Rule~\ref{rule:1}, $u$ remains an internal
  vertex as $d_{T'}(u) = d_T(u) - 1\ge 2$.  Note that possibly $v\in
  \{l_1, l_2\}$, and in this case $v$ remains a leaf but the other
  leaf becomes an internal vertex.  Otherwise, both $l_1$ and $l_2$
  become internal vertices, while $v$ might become a leaf.  In either
  case, the number of internal vertices increases by at least $1$.

  We now consider Rule~\ref{rule:2}.  Case \eqref{case2} is
  straightforward: after its application, all vertices in $\{l, w, u,
  v\}$ are internal vertices of the resulting tree $T'$; hence
  $|I(T')| = |I(T)| + 1$.  In case \eqref{case3}, after the
  substitution, $w$ and $u$ remain internal vertices, while $l$
  becomes an internal vertex and $v$ becomes a leaf of the new tree.
  Albeit the number of leaves does not change, the cotree edge of the
  new tree between $l_v$ and the new leaf $v$ enables us to apply
  Rule~\ref{rule:1}, which increases the number of internal vertices
  by $1$.  Case \eqref{case4} is similar as case \eqref{case3}: the
  first substitution does not affect the number of internal vertices,
  but the subsequent application of Rule~\ref{rule:1} increases it by
  $1$.  In case \eqref{case5}, after the first substitution, $w$
  remains an internal vertex, $l$ becomes an internal vertex, while
  both $u$ and $v$ become leaves; hence the number of leaves decreases
  by $1$.  Since $l_u$, $u$, $l_v$, $v$ are four distinct leaves in
  the resulting tree, applying Rule~\ref{rule:1} to $l_u u$ results in
  a spanning tree of $G$ that has at least $I(T)$ internal vertices.
  Either $l_v$ and $v$ (which remain leaves) are the only leaves and
  we are done, or we can apply Rule~\ref{rule:1} to $l_v v$ to
  increase the number of internal vertices.  This obtained tree $T'$
  has at least one more internal vertex than $T$.
\end{proof}

A spanning tree $T$ of a graph $G$ is called {\em maximal} if neither
exchange rule is applicable to it.  By Lemma~\ref{lem0}, each
application of an exchange rule to a spanning tree increases its
number of internal vertices at least $1$.  Since a spanning tree of
$G$ has at most $|V(G)| - 2$ internal vertices, we have the following
theorem.
\begin{theorem}\label{thm:maximal-trees}
  A maximal spanning tree of a graph can be constructed in polynomial
  time.
\end{theorem}

\section{The kernelization algorithm}
Although neither exchange rule reduces the graph size, their
exhaustive application provides a maximal spanning tree whose
structural properties will be crucial for our kernelization algorithm.
In this section we are only concerned with maximal spanning trees.  We
will use the reduction rule of Fomin et
al.~\cite{fomin-13-kernel-mist}, which is recalled below.  Let \opt{G}
denote the maximum number of internal vertices a spanning tree of $G$
can have.
%
\begin{lemma}[\cite{fomin-13-kernel-mist}]
  \label{lem8} 
  Let $L'$ be an independent set of $G$ such that $|L'| \ge 2
  |N_G(L')|$.  We can find in polynomial time nonempty subsets
  $S\subseteq N_G(L')$ and $L\subseteq L'$ such that:
  \begin{enumerate}[(1)]
  \item $N_G(L) = S$, and
  \item the graph $(S\cup L, E(G)\cap (S \times L))$ has a spanning
    tree such that all vertices of $S$ and $|S|-1$ vertices of $L$ are
    internal.
  \end{enumerate}
  Moreover, let $G'$ be obtained from $G$ by adding a vertex $v_{\mbox
    {\tiny $S$}}$ adjacent to every vertex in $N_G(S)\backslash L$,
  adding a vertex $v_{\mbox {\tiny $L$}}$ adjacent to $v_{\mbox {\tiny
      $S$}}$, and removing all vertices of $S\cup L$, then $\opt{G'} =
  \opt{G} - 2|S| + 2$.
\end{lemma}
Note that $|L| \ge 2$ (otherwise the graph in Lemma~\ref{lem8}(2)
cannot have internal vertices), and hence each application of the
reduction rule decreases the number of vertices by at least $1$.  The
safeness of the following reduction rule is ensured by
Lemma~\ref{lem8}.

\paragraph{Reduction Rule} (\cite{fomin-13-kernel-mist}): Find
nonempty subsets $S$ and $L$ of vertices as in Lemma~\ref{lem8}.
Return ($G', k'$) where $G'$ is defined in Lemma~\ref{lem8} and $k' =
k - 2|S| + 2$.

\medskip

The main technical obstacle is then to identify a vertex set $L'$ with
$|L'| \ge 2 |N_G(L')|$.  The is trivial when $|V(G)| \ge 3 k - 3$.  In
any spanning tree $T$ of $G$ with $|L(T)| < k$ it holds that $|L(T)|
\ge 2k - 2 \ge 2 |I(T)| = 2 |N_G(L(T))|$; hence we can use $L(T)$ as
$L'$ and a $3 k$-vertex kernel follows.  However, it becomes very
nontrivial to find such a set when $2k < |V(G)| < 3 k - 3$.  Our
approach here is to separate a maximal spanning tree $T$ into several
subtrees and bound the number of $L(T)$ by the number of $I(T)$
residing in each subtree individually.  It is worth mentioning that a
leaf of a subtree may not be a leaf of $T$.

Recall that the removal of any edge $u v\in E(T)$ from a tree $T$
breaks it into two components, one containing $u$ and the other
containing $v$.  In general, the removal of all edges of an edge
subset $E'\subseteq E(T)$ from $T$ breaks it into $|E'| + 1$
components, each being a subtree of $T$.  We would like to divide $T$
in a way that the two ends of any good cotree edge always reside in
the same subtree, hence the following definition.
\begin{definition} 
  An edge $u v\in E(T)$ connecting two internal vertices $u,v$ of $T$
  is {\em critical} if there is no good cotree edge connecting the two
  components of $T - uv$.
\end{definition}
Let $C(T)$ denote the (possibly empty) set of all critical edges in
$T$.  Note that for each non-critical edge $u v$ with both $u, v \in
D(T)$, there must be a good cotree edge connecting the two components
of $T - uv$.
\begin{lemma}\label{lem1} 
  Let $u$ and $v$ be two detachable vertices of a maximal spanning
  tree $T$ and $u v\in E(T)$.  If there is a good cotree edge
  connecting the two components of $T - uv$, then there exists a good
  cotree edge $l w$ such that
  \begin{enumerate}[(1)]
  \item $u\in D_2(T)$, the only good cotree edge incident to $u$ is $l
    u$, and $l w$ crosses $u v$; or
  \item $v\in D_2(T)$, the only good cotree edge incident to $v$ is $l
    v$, and $l w$ crosses $v u$.
  \end{enumerate}
\end{lemma}
\begin{proof}
  Let $l' w'$ be a good cotree edge connecting the two components of
  $T - uv$.  Without loss of generality, we may assume that $l' w'$
  cross $u v$---then $l'$ and $w'$ are in the components of $T - uv$
  containing $u$ and $v$ respectively---the other case follows by
  symmetry.

  We argue first that $u\in D_2(T)$.  Suppose for contradiction, $u\in
  I_3(T)$.  Since Rule~\ref{rule:2}\eqref{case2} is not applicable to
  $l' w'$ and $u v$, we must have $v\not\in I_3(T)$.  Then $v\in
  D_2(T)$, and by definition, there is a good cotree edge $l_v v$ of
  $T$.  Again, since Rule~\ref{rule:2}\eqref{case2} is not applicable
  to $l_v v$ and $u v$, we must have $l_v\ne l'$.  However, $l_v \neq
  l'$ and $v \neq w$ would imply that Rule~\ref{rule:2}\eqref{case3}
  is applicable to $l' w'$, $u v$, and $l_v v$, a contradiction.  Now
  that $u\in D_2(T)$, by definition, there is a good cotree edge
  incident to $u$.  The proof is now completed if $l' u$ is a good
  cotree edge of $T$ and the only one incident to $u$: we are in case
  (1) and $l' w'$ is the claimed edge.  In the rest of the proof we
  may assume otherwise---that is, there is a good cotree edge $l u$ of
  $T$ with $l \ne l'$.

  Since Rule~\ref{rule:2}\eqref{case4} cannot be applicable to $l'
  w'$, $u v$, and $l u$, the vertex $v$ is also in $D_2(T)$ and $v
  \neq w'$.  There is also a good cotree edge $l_v v$ of $T$.
  Similarly it can be inferred that $l_v\ne l'$: otherwise
  Rule~\ref{rule:2}\eqref{case4} is applicable to $l' v$ (i.e., $l_v
  v$), $u v$, and $l u$.  Now that $l'$ is different from both $l$ and
  $l_v$ but Rule~\ref{rule:2}\eqref{case5} is not applicable to $l'
  w'$, $u v$, $l u$ and $l_v v$, we must have $l = l_v$, i.e., both $l
  u$ and $l v$ are good cotree edges of $T$.  We consider now which
  component of $T - uv$ the leaf $l$ belongs to.  If it is with $u$,
  then there cannot be another good cotree edge $l'_u u$ with $l'_u
  \neq l$: otherwise, Rule~\ref{rule:2}\eqref{case4} would be
  applicable on $l v$, $u v$, and $l'_u u$.  In other words, $l u$ is
  the only good cotree edge incident to $u$, and $l v$ is the claimed
  good cotree edge crossing $u v$.  We are in case (1) and $l v$ is
  the claimed edge.  A symmetric argument implies that we are in case
  (2) and $l u$ is the claimed edge if $l$ is in the component of $T -
  uv$ with $v$.
  This concludes the proof.
\end{proof}
The vertices $u$ or $v$ stipulated in Lemma \ref{lem1} turns out to be
our main trouble in analyzing the size of reduced instance.  We use
$D_B(T)$ to denote this set of vertices, which need our special
attention.  For the pair of vertices $u,v$ as in Lemma~\ref{lem1},
$u\in D_B(T)$ if case (1) holds true, and $v\in D_B(T)$ otherwise.
Note that $D_B(T) \subseteq D_2(T) \subseteq I_2(T)$.  This vertex has
degree $2$ in $T$, and its other neighbor (different from $v$ or $u$)
cannot be a leaf of $T$: suppose that it is $u$; the component of $T -
uv$ containing $u$ must contain another $l\in L(T)$ that is
nonadjacent to $u$ in $T$.  Therefore, the only good cotree edge
mentioned in Lemma \ref{lem1} is actually the only edge between it and
$L(T)$ in $G$.  The following corollary follows easily.
\begin{corollary}\label{cor1.5} 
  For each $u\in D_B(T)$ for a maximal spanning tree $T$, we have
  $|L(T)\cap N_G(u)| = 1$.
\end{corollary}
Note that $I_3(T) \subseteq D(T)\setminus D_B(T)$ and $I(T)\setminus
D(T)\subseteq I_2(T)$.

\begin{lemma}\label{lem6} 
  Let $T$ be a maximal spanning tree.  For every pair of vertices $u,
  w\in D(T)\setminus D_B(T)$ that are in the same component of $T -
  C(T)$, the path $P_T(u, w)$ visits at least one vertex $v\in
  I(T)\setminus D(T)$ such that for any $l \in L(T)$, the path $P_T(l,
  v)$ visits $D(T)$.
\end{lemma}
\begin{proof} 
  No generality will be lost by assuming that $P_T(u, w)$ is minimal
  (in the sense that it visits no other vertex in $D(T)\setminus
  D_B(T)$).  By assumption, $P_T(u, w)$ is retained in $T - C(T)$.  We
  argue first $u w\not\in E(T)$.  Suppose for contradiction, $u w\in
  E(T)$.  Since $u v\not\in C(T)$, there must be some good cotree edge
  crossing it; however, by Lemma~\ref{lem1} and the definition of
  $D_B(T)$, at least one of $u$ and $w$ is then in $D_B(T)$, a
  contradiction.  Let $P_T(u, w) = u v_1 \cdots v_p w$, where $p \ge
  1$; note that $d_T(v_i) = 2$ and $v_i\not\in D(T)\setminus D_B(T)$
  for each $1\le i\le p$.

  We now find an internal vertex of $P_T(u, w)$ that is not in $D(T)$
  as follows.  If $v_1\not\in D(T)$, then we are done.  Otherwise,
  $v_1\in D_B(T)$, and there is a unique good cotree edge $l v_1$.  We
  prove by contradiction that $l$ is in the same component of $T - u
  v_1$ with $v_1$.  Suppose the contrary, then
  \begin{itemize}
  \item if $u\in I_3(T)$, then Rule~\ref{rule:2}\eqref{case2} is
    applicable to $l v_1$ and $u v_1$; or
  \item if $u\in D_2(T)\setminus D_B(T)$, then there is a good cotree
    edge $l' u$ with $l'\neq l$, and hence
    Rule~\ref{rule:2}\eqref{case4} is applicable to $l v_1$, $u v_1$,
    and $l' u$.
  \end{itemize}
  Noting that every internal vertex of $P$ has degree $2$, this
  actually implies that $l$ must be in the same component of $T - v_p
  w$ with $w$.  The first $v_i$ such that $l v_i$ is {\em not} the
  only good cotree edge incident to $v_i$ is the vertex we need.  Its
  existence can be argued by contradiction as follows.  Suppose for
  contradiction, there is no such a vertex $v_i$, then $l v_p$ is the
  only good cotree edge incident to $v_p$, and
  \begin{itemize}
  \item if $w\in I_3(T)$, then Rule~\ref{rule:2}\eqref{case2} is
    applicable to $l v_p$ and $w v_p$; or
  \item if $w\in D_2(T)\setminus D_B(T)$, then there is a good cotree
    edge $l' w$ with $l'\neq l_1$, and hence
    Rule~\ref{rule:2}\eqref{case4} is applicable to $l v_p$, $w
    v_p$, and $l' w$.
  \end{itemize}
  These contradictions imply that $p \ge 2$ and there must be $2\le
  i\le p$ such that $v_i$ is incident to a good cotree edge $l' v_i$
  with $l' \ne l$.  Since Rule~\ref{rule:2}\eqref{case4} is not
  applicable to $l v_{i-1}$, $v_i v_{i-1}$, and $l' v_i$, we can
  conclude that $v_i\not\in D_2(T)$.  Noting that $v_i\in I_2(T)$, we
  have verified that $v_i$ is an internal vertex of $P_T(u,w)$ not in
  $D(T)$.

  Noting that every internal vertex of $P_T(u,w)$ has degree $2$, for
  any $l \in L(T)$, the path $P_T(l,v)$ necessarily visits either $u$
  or $w$, which is in $D(T)$.  This concludes the proof.
\end{proof}

We are now ready for proving the main result of this section.
\begin{lemma}\label{lem:reduction}
  Let $T$ be a maximal spanning tree of $G$ with $|V(G)|\ge 4$.  If
  $|L(T)| > |I(T)|$, then we can find in polynomial time an
  independent set $L'$ of $G$ such that $|L'| \ge 2|N_G(L')|$.
\end{lemma}
\begin{proof}
  We find all critical edges $C(T)$, and take the forest $T - C(T)$.
  By assumption, there must be some component $T_0$ of $T - C(T)$ of
  which more than half vertices are from $L(T)$.  Let $X$ and $Y$
  denote $L(T)\cap V(T_0)$ and $I(T)\cap V(T_0)$ respectively; then
  $|X| \ge |Y| + 1$.  Since $T$ is maximal (Rule~\ref{rule:1} is not
  applicable), $X$ is an independent set.  We divide $X$ into the
  following three subsets:
  \[
  X_1 := X\cap N_G(D_B(T));\qquad X_2 := X\cap N_T(D(T)) \setminus X_1;
  \quad \text{ and }\quad X_3 := X\setminus (X_1\cup X_2).
  \]
  We will show that $|X_2| \ge 2 |N_G(X_2)|$, and hence $X_2$
  satisfies the claimed condition and can be used as $L'$.  By the
  definition of critical edges, there is no good cotree edge of $T$
  connecting two different components of $T - C(T)$; hence
  $N_G(X)\subseteq Y$.  We accordingly divide $Y$ into subsets.  The
  detachable vertices are either in $Y_1 := D_B(T)\cap V(T_0)$ or $Y_2
  := \left( D(T)\setminus D_B(T) \right) \cap V(T_0)$, while a vertex
  $y\in Y\setminus D(T)$ is in $Y_3$ if there exists $l\in L(T)$ such
  that the path $P_T(l,y)$ does not visit $D(T)$, or in $Y_4$
  otherwise.  Note that $|X| = |X_1| + |X_2| + |X_3|$ and $|Y| =
  |Y_1| + |Y_2| + |Y_3| + |Y_4|$.

  We argue first that $N_G(X_2)\subseteq Y_2$.  It suffices to show
  $N_G(X_2)\subseteq D(T)$ (the definition of $X_2$ requires that a
  vertex in it is nonadjacent to $D_B(T)$ in $G$), which further boils
  down to showing $N_G(X_2)\cap I_2(T)\subseteq D(T)$: since $T$ is
  maximal (Rule~\ref{rule:1} is not applicable), $X_2$ has no neighbor
  in $L(T)$; on the other hand, $I_3(T)\subseteq D(T)$.  Consider a
  vertex $x \in X_2$, and let $y$ be the unique neighbor of $x$ in
  $T$.  By assumption, $y\in D_2(T)\setminus D_B(T)$, and hence there
  is a good cotree edge $l y$ of $T$ with $l\neq x$.  For each $y'\in
  N_G(x)\cap I_2(T)$ different from $y$, the path $P_T(x, y')$ visits
  $y$, using the definition of $D_2(T)$ we can conclude that $y'\in
  D_2(T)$.

  Each $x\in X_1$ has a neighbor $y \in Y_1$.  By
  Corollary~\ref{cor1.5}, $x$ is the only vertex in $N_G(y) \cap
  L(T)$.  Thus, $|X_1| \le |Y_1|$.  The unique neighbor $y$ of a
  vertex $x\in X_3$ in $T$ must be in $I_2(T)\setminus D_2(T)$.  Since
  the trivial path $P_T(x, y)$ (consisting of a single edge $xy$) does
  not visit $D(T)$, we have $y\in Y_3$.  The other neighbor of $y$ in
  $T$ cannot be a leaf of $T$ ($G$ has at least four vertices).  Thus,
  $|X_3| \le |Y_3|$.  By Lemma~\ref{lem6}, for any two different
  vertices $u$ and $w$ of $Y_2$, the path $P_T(u,w)$ visits at least
  one vertex in $Y_4$.  Since $T_0$ is a tree, using induction it is
  easy to show $|Y_4| \ge |Y_2| - 1$.

  Summarizing above, we have
  \begin{align*}
    |X_2| &= |X| - |X_1| - |X_3| & \qquad (\text{because }|X| = |X_1|
    + |X_2| + |X_3|.)
    \\
    &\ge |Y| + 1 - |Y_1| - |Y_3| & \qquad (\text{because }|X| \ge |Y|
    + 1; |X_1| \le |Y_1|; |X_3| \le |Y_3|.)
    \\
    &= |Y_2| + |Y_4| + 1 & \qquad (\text{because }|Y| = |Y_1| + |Y_2|
    + |Y_3| + |Y_4|.)
    \\
    &\ge 2|Y_2| & \qquad (\text{because }|Y_4| \ge |Y_2| - 1.)
    \\
    &\ge 2|N_G(X_2)|.& \qquad (\text{because }N_G(X_2)\subseteq Y_2.)
  \end{align*}
  Hence $X_2$ can be used as the independent set $L'$.  This concludes
  the proof.
\end{proof}
Lemmas~\ref{lem:reduction} and \ref{lem8}, together with
Theorem~\ref{thm:maximal-trees}, imply Theorem~\ref{thm:main}.

\end{document}